\documentclass[showpacs,amsmath,amssymb,twocolumn,superscriptaddress,notitlepage,preprintnumbers,pra]{revtex4-2}
\usepackage{color}
\usepackage{physics}
\usepackage{graphicx}
\usepackage{dcolumn}%

\usepackage{float}
\usepackage{algorithm}
\usepackage{algpseudocode}

\makeatletter
\renewenvironment{algorithm}[1][htb]{%
    \refstepcounter{algorithm}%
    \par\medskip\noindent%
    \hrule height 0.8pt%
    \kern2pt%
    \textbf{Algorithm~\thealgorithm. #1}%
    \kern2pt\hrule height 0.8pt%
    \kern4pt%
}{%
    \kern4pt\hrule height 0.8pt\par\medskip
}
\makeatother

\usepackage{hyperref}
\usepackage{subfigure}
\hypersetup{colorlinks=true, linkcolor=blue, citecolor=blue, urlcolor=black}
\usepackage{braket}

\usepackage{amsmath, amsthm, amssymb, bm}

\theoremstyle{plain}
\newtheorem{theorem}{Theorem}
\newtheorem{lemma}{Lemma}
\newtheorem{corollary}{Corollary}

\theoremstyle{definition}
\newtheorem{definition}{Definition}

\theoremstyle{remark}

\newcommand{\proofstep}[1]{\medskip\noindent\textbf{#1}}

\definecolor{hjx}{rgb}{0.4,0.3,0.9}

\begin{document}


\title{Estimating Local Observables via Cluster-Level Light-Cone Decomposition}

\author{Junxiang Huang}
\affiliation{School of Computer Science, Peking University, Beijing 100871, China}
\affiliation{Center on Frontiers of Computing Studies, Peking University, Beijing 100871, China}

\author{Yunxin Tang}
\affiliation{Center on Frontiers of Computing Studies, Peking University, Beijing 100871, China}
\affiliation{School of Electronics Engineering and Computer Science, Peking University, Beijing 100871, China}

\author{Xiao Yuan}
\email{xiaoyuan@pku.edu.cn}
\affiliation{Center on Frontiers of Computing Studies, Peking University, Beijing 100871, China}

\date{\today}
\begin{abstract}
Simulating large quantum circuits on hardware with limited qubit counts is often attempted through methods like circuit knitting, which typically incur sample costs that grow exponentially with the number of connections cut. 
In this work, we introduce a framework based on Cluster-level Light-cone analysis that leverages the natural locality of quantum workloads. We propose two complementary algorithms: the {Causal Decoupling Algorithm}, which exploits geometric disconnections in the light cone for sampling efficiency, and the {Algebraic Decomposition Algorithm}, which utilizes algebraic expansion to minimize hardware requirements. These methods allow simulation costs to depend on circuit depth and connectivity rather than system size.
Together, our results generalize Lieb-Robinson-inspired locality to modular architectures and establish a quantitative framework for probing local physics on near-term quantum devices by decoupling the simulation cost from the global system size.
\end{abstract}
\maketitle

\section{Introduction}

The noisy intermediate-scale quantum (NISQ) era is characterized by a significant gap between the ambitious goals of quantum algorithms and the practical limitations of current hardware~\cite{preskill2018quantum,daley2022practical}. Today's leading platforms, such as superconducting circuits and trapped-ion systems, are constrained by modest qubit counts, finite coherence times, and restricted connectivity, with gate fidelities still far from the fault-tolerant threshold~\cite{terhal2015quantum}. This hardware reality makes the direct execution of large-scale quantum algorithms on a single monolithic processor infeasible. Consequently, a primary challenge in quantum computing research is to develop hybrid quantum-classical strategies that can effectively leverage these small, imperfect devices to tackle problems of practical significance~\cite{peruzzo2014variational,farhi2014quantum}.

A long-term vision to overcome this barrier is the construction of distributed, modular quantum computers: multiple high-precision quantum cores linked by quantum networks for true scalability~\cite{kimble2008quantum,wehner2018quantum}. 
However, until reliable quantum interconnects are widely available, hybrid schemes based on local operations and classical communication (LOCC) offer a stopgap. 
By either parallelizing subcircuits across distinct processors or time-slicing larger algorithms on one device, LOCC protocols can stitch together multiple small quantum units into an effectively larger machine~\cite{bravyi2018quantum,peng2020simulating}.

Meanwhile, different divide-and-conquer strategies, such as \textit{circuit cutting/knitting}, \textit{quasi-probability decompositions}, and classical perturbative splitting of Hamiltonians, have been proposed recently. 
\textit{Circuit cutting} partitions cross-block gates at the cost of exponential classical post-processing once more than a handful of cuts are introduced~\cite{peng2020simulating,tang2021cutqc,tang2022scaleqc,perlin2021quantum,lowe2023fast,harada2024doubly, ufrecht2023cutting,ufrecht2024optimal,schmitt2025cutting}. 
\textit{Quasi-probability} methods simulate non-physical channels via ``probabilities'' including even negative ones, causing variance to blow up with circuit size~\cite{Hofmann_2009,bravyi2016trading,mitarai2021constructing,piveteau2022quasiprobability,carrera2024combining}. 
Similarly, perturbative splitting of weak Hamiltonian terms transfers complexity into a classical expansion whose number of samples grows exponentially in the strength of total interactions~\cite{yuan2021quantum,fujii2022deep,sun2022perturbative,barratt2021parallel,gentinetta2024overhead,harrow2025optimal}. 
Together, these approaches incur exponential sampling overhead in the number of cuts under standard assumptions, which in practice limits the feasible number of cuts.~\cite{mitarai2021overhead,piveteau2023circuit,jing2025circuit}.

Therefore, developing a simulation strategy that mitigates this system-size-exponential scaling barrier is a critical challenge for leveraging modular quantum architectures.

In this work, we introduce and analyze two complementary frameworks based on a Cluster-level Light-cone analysis, which leverage the natural locality of many quantum workloads~\cite{lieb1972finite,georgescu2014quantum}. 
Our first approach is the {Causal Decoupling Algorithm}, a geometric method that exploits the potential causal disconnection of the light cone to factorize the quantum simulation into independent sub-problems. This strategy offers high sampling efficiency but requires hardware capable of simulating connected light-cone components.
Our second approach, the {Algebraic Decomposition Algorithm}, further expands this framework algebraically. It transforms the entangled evolution within a light cone into a sum of local, unentangled operations. While the number of terms in this sum is exponential in the light cone volume (and thus circuit depth), it is crucially independent of the total system size, thereby avoiding the exponential scaling barrier in the number of qubits. For shallow-depth circuits, this classically tractable expansion allows the full expectation value to be reconstructed from measurements on minimal, single-cluster quantum devices.

Our contributions are as follows:
\begin{enumerate}
  \item We present the Causal Decoupling Algorithm as a geometric decomposition strategy, which leverages the causal structure of the light cone to achieve polynomial sample complexity.
  \item We formalize the Algebraic Decomposition framework, demonstrating how to convert complex entangled evolution into a sum of local tasks, enabling simulation on minimal hardware at the cost of depth-dependent sampling overhead.
  \item For both frameworks, we derive rigorous complexity bounds demonstrating their dependence on circuit depth and local connectivity, rather than the total system size.
\end{enumerate}

The rest of this paper is organized as follows. Section~\ref{sec:bg} reviews the necessary background. Section~\ref{sec:main} presents our main results, detailing both the Causal Decoupling Algorithm and the Algebraic Decomposition Algorithm, along with their corresponding theoretical performance bounds. We discuss the implications and compare the frameworks in Section~\ref{sec:discuss}. All detailed proofs for our theorems and their corollaries are deferred to the Appendices.

\section{Background}
\label{sec:bg}

A critical task in many quantum algorithms is the estimation of expectation values of observables. Specifically, one often needs to compute the quantity 
\begin{equation}
    \mu = \langle 0 | U^\dagger\,O\,U | 0 \rangle,
\end{equation}
where \(U\) is a unitary transformation applied to an initial state \(|0\rangle\). The decomposition of \(O\) into a sum of Pauli operators further complicates this task \cite{knill2007optimal,huang2020classical}. In practical scenarios, the operator \(O\) is commonly expressed as a weighted sum of tensor-product Pauli terms:
\begin{equation}
    O = \sum_{\alpha=1}^m c_\alpha \left(P_{\alpha1} \otimes \cdots \otimes P_{\alpha n}\right),
\end{equation}
where each \(P_{\alpha j}\) is one of the single-qubit Pauli operators \(\{I,X,Y,Z\}\) \cite{bravyi2021classical}. The structure of these Pauli terms plays a crucial role in the computational complexity of estimating \(\mu\), as it determines how quantum correlations propagate through the system \cite{huang2021derandomization}. 

We introduce a framework to quantify the locality of such observables at both the qubit and cluster levels. 
A Pauli term is said to be \emph{\(s\)-local} if it acts nontrivially on at most \(s\) qubits (or clusters). 
This definition enables a formal characterization of the causal propagation of quantum information, 
analogous to the Lieb--Robinson bounds on information spreading in local quantum systems~\cite{lieb1972finite,bravyi2006liebh}. 
At the cluster level, an \(s\)-local observable interacts with at most \(s\) clusters. 

In geometrically structured quantum systems, such as those used in quantum many-body simulations and quantum error correction, subsystems are arranged spatially so that interactions occur primarily between neighboring qubits \cite{hastings2019locality}. This locality can be exploited to design more efficient simulation algorithms, especially for shallow circuits in which quantum information spreads only to a limited extent (finite light-cone effect) \cite{sun2022perturbative,singh2023experimental,baumer2023efficient}. By extending the concepts of light cones and causal propagation to the cluster level, we can model how information flows within and between clusters, thereby reducing the computational complexity of estimating expectation values \cite{bravyi2006liebh}. 

\subsection{Causal Structure at the Qubit Level}
In spatially structured quantum circuits, the spread of quantum correlations is limited by both the circuit depth and the underlying geometry of qubit interactions. To rigorously characterize this causal propagation, we begin with the standard definitions at the qubit level.

\begin{definition}[Qubit Light Cone]
Let the unitary evolution of a quantum circuit be given by
\begin{equation}
    U = \prod_{t=1}^T V_t,
\end{equation}
where each $V_t$ denotes the set of local quantum gates applied sequentially at time step $t$, and $T$ is the total number of circuit layers (i.e., the circuit depth).

For a target qubit $j$, we denote by $\mathcal{H}_j \cong \mathbb{C}^2$ its single-qubit Hilbert space, and by $\mathcal{O}(\mathcal{H}_j)$ the algebra of linear operators acting on $\mathcal{H}_j$ (i.e., all possible local observables on qubit $j$).

The \emph{light cone} $L(j) \subseteq \{1,\ldots,n\}$ of qubit $j$ is defined as the minimal set of qubits satisfying
\begin{equation}
    \forall\, O_j \in \mathcal{O}(\mathcal{H}_j), \quad
    U^\dagger O_j U = U_{\mathrm{loc}}^\dagger(j)\, O_j\, U_{\mathrm{loc}}(j),
\end{equation}
with
\begin{equation}
    U_{\mathrm{loc}}(j) = \prod_{V_t \in \mathcal{G}(L(j))} V_t,
\end{equation}
where $\mathcal{G}(L(j))$ denotes the set of gates whose support is entirely contained within $L(j)$.
\end{definition}
This definition encapsulates the idea that, under the full evolution $U$, the Heisenberg picture transformation of any local observable $O_j$ depends only on the operations within its light cone; contributions from gates outside $L(j)$ cancel out due to a cancellation mechanism.

\begin{definition}[Qubit-Level Range on Lattice Systems]
\label{def:qubit_range}
For a quantum circuit where qubits are arranged in a specific geometric layout (such as a 1D line or a 2D grid), which provides a meaningful distance metric, we can define the geometric range of a qubit's light cone.

For a target qubit $j$ with light cone $L(j)$, its \emph{range} is defined as the maximum distance from itself to any other qubit within its light cone:
\begin{equation}
    r(j) = \max_{k \in L(j)} d(j,k),    
\end{equation}
where $d(j,k)$ represents the shortest path distance between qubits $j$ and $k$ on the circuit's interaction graph or geometric layout.

The \emph{causal range of the unitary} $U$ is then defined as the maximum such range across all qubits:
\begin{equation}
    r(U) = \max_j r(j).
\end{equation}
\end{definition}

The range quantifies the maximum spatial extent over which local operations can causally affect the observable at qubit $j$.

These definitions formalize the notion that in shallow, geometrically organized circuits, the evolution of a local observable is influenced predominantly by a confined subset of qubits. Due to the limited depth $D$ of such circuits, Lieb-Robinson type bounds ensure that the range $r(j)$ scales at most linearly with $D$~\cite{hastings2004locality}, thereby localizing quantum correlations. This intrinsic locality is central to the development of efficient simulation strategies, as it permits the classical reconstruction of global circuit behavior from the dynamics within localized regions. Such approaches are crucial for simulating large-scale, geometrized subsystem quantum circuits on devices with limited quantum resources.

\subsection{Clustered Quantum Architectures}
We now extend these concepts to clustered quantum systems, which better reflect the architecture of many modern quantum devices~\cite{monroe2014large}. In such systems, qubits are partitioned into clusters, and information propagates via inter-cluster operations.

\begin{definition}[Clustered Quantum System]
\label{def:clustered_system}
An \emph{\(n\)-qubit \(N\)-cluster clustered quantum system} consists of an $n$-qubit Hilbert space $\mathcal{H} = (\mathbb{C}^2)^{\otimes n}$ that is partitioned into a set of $N$ disjoint regions called \emph{clusters}, $\mathcal{C} = \{C_1, \dots, C_N\}$. Each cluster $C_j$ is a set of qubits, such that:
\begin{enumerate}
    \item The clusters form a partition of the $n$ qubits: $\bigcup_{j=1}^N C_j = \{1, \dots, n\}$ and $C_j \cap C_k = \emptyset$ for $j \neq k$.
    \item Each cluster $C_j$ contains at most $d$ qubits, i.e., $|C_j| \le d$.
\end{enumerate}
\end{definition}

\begin{definition}[$D$-Dimensional Clustered Lattice]
\label{def:d_dim_lattice}
A \emph{$D$-dimensional clustered lattice} is a specific instance of a clustered quantum system (Def.~\ref{def:clustered_system}) where the clusters are arranged on the vertices of a $D$-dimensional cubic lattice graph, $\mathbb{Z}^D$. 
A cluster $C_{\mathbf{i}}$ is identified by a coordinate vector $\mathbf{i} = (i_1, \dots, i_D) \in \mathbb{Z}^D$. Two clusters $C_{\mathbf{i}}$ and $C_{\mathbf{j}}$ are considered \emph{adjacent} or nearest-neighbors if their coordinates differ by one in a single dimension, i.e., $\|\mathbf{i} - \mathbf{j}\|_1 = \sum_{k=1}^D |i_k - j_k| = 1$. Physical interactions between qubits in different clusters are restricted to occur only between adjacent clusters.
\end{definition}

\begin{definition}[All-to-All Connected Clustered System]
\label{def:all_to_all_system}
An \emph{all-to-all connected clustered system} is a specific instance of a clustered quantum system (Def.~\ref{def:clustered_system}) where the inter-cluster interaction graph is a complete graph.

In contrast to the geometrically constrained lattice, this architecture permits physical interactions, such as two-qubit gates, to act between any pair of distinct clusters $(C_i, C_j)$ in the system. The application of inter-cluster gates is not restricted by any notion of adjacency or spatial proximity. This models scenarios with long-range interactions or systems where logical connectivity is fully decoupled from the physical layout, such as in certain trapped-ion or photonic platforms.
\end{definition}

\subsection{Light Cones and Observables in Clustered Systems}
Having defined the physical architectures, we can now formalize the concept of causal propagation and define the properties of observables at the cluster level.

\begin{definition}[Cluster Light Cone]
Let the quantum system be an $n$-qubit clustered quantum system partitioned into $N$ clusters. Consider an arbitrary quantum circuit given by
\begin{equation}
    U = \prod_{t=1}^T \left( W_t \cdot V_t \right),
    \label{eq:cluster_U}
\end{equation}
where:
\begin{itemize}
    \item \(V_t = \bigotimes_{(C_j,C_k) \in E_t} V_{jk}^t\) represents the \emph{inter-cluster operations} at layer \(t\), applying local gates only on geometrically adjacent cluster pairs as defined by \(E_t\) (for instance, nearest-neighbor clusters in a two-dimensional grid).
    \item \(W_t = \bigotimes_{C_j} W_{C_j}^t\) represents the \emph{intra-cluster operations} at layer \(t\), which permit arbitrary quantum gates within a single cluster \(C_k\) without any geometric locality constraints.
\end{itemize}
For a target cluster \(C_j\), its \emph{cluster light cone} \(\mathcal{L}(C_j) \subseteq \{1,\dots,N\}\) is defined as the minimal set of clusters satisfying
\begin{equation}
    \forall\, O^{(C_j)} \in \mathcal{O}\left(\mathcal{H}_{C_j}\right),\quad U^\dagger O^{(C_j)} U = \widetilde{U}(j)^\dagger O^{(C_j)} \widetilde{U}(j),
\end{equation}
where
\begin{equation}
    \widetilde{U}(j) = \prod_{t=1}^T \left( \widetilde{W}_t(j) \cdot \widetilde{V}_t(j) \right),
\end{equation}
with 
\begin{equation}
    \widetilde{W}_t(j) = \bigotimes_{C_k \in \mathcal{L}(C_j)} W_{C_k}^t,
\end{equation}
which retains only those intra-cluster operations acting on clusters within \(\mathcal{L}(C_j)\), and \(\widetilde{V}_t(j)\) is the truncated inter-cluster operation that includes only those gates in \(V_t\) whose endpoints both belong to \(\mathcal{L}(C_j)\).

\end{definition}

Having defined the cluster light cone $\mathcal{L}(C_j)$ as the set of all causally relevant clusters, we now introduce two key metrics to quantify its properties. These metrics are crucial as they directly determine the complexity of our simulation algorithms. The first metric, volume, quantifies the total amount of inter-cluster interaction within the cone, which is the primary cost factor for algebraic decomposition methods. The second metric, range, measures the geometric extent of the cone, which is particularly relevant for systems on structured lattices.

\begin{definition}[Volume of the Light Cone Circuit]
\label{def:volume_light_cone}
For a target cluster $C_j$, let $\mathcal{L}(C_j)$ be its cluster light cone and let 
\begin{equation}
    \widetilde{U}(j) = \prod_{t=1}^T \left( \widetilde{W}_t(j) \cdot \widetilde{V}_t(j) \right)
\end{equation}
be the corresponding \emph{light cone circuit}, which is the minimal sub-circuit of $U$ that reproduces the dynamics of any observable initially supported on $C_j$.

The volume of the light cone circuit, denoted as $\text{Vol}(\mathcal{L}(C_j))$, is defined as the total count of fundamental two-qubit gates that constitute the truncated inter-cluster operations $\widetilde{V}_t(j)$ across all time steps.

Specifically, if we let $N(\widetilde{V}_t(j))$ be the number of two-qubit gates in the operator $\widetilde{V}_t(j)$ at time step $t$, the volume is given by:
\begin{equation}
    \text{Vol}(\mathcal{L}(C_j)) := \sum_{t=1}^T N(\widetilde{V}_t(j)).
    \label{eq:light_cone_volume}
\end{equation}
\end{definition}

This volume quantifies the total interaction budget within the causal past, directly impacting the simulation complexity of algebraic decomposition methods.

\begin{definition}[Cluster-Level Range on Lattice Clustered Systems]
\label{def:cluster_range}
For a cluster-level system defined on a geometric lattice, where a meaningful distance metric between clusters exists, we can define the geometric range of the causal light cone.

Let $\delta_{\text{inter}}(C_j,C_k)$ be the inter-cluster geometric distance between clusters $C_j$ and $C_k$ on the lattice (e.g., the shortest path length in the cluster connectivity graph). The \emph{cluster-level range} $\mathcal{R}(C_j)$ of a target cluster $C_j$ is the maximum geometric distance from $C_j$ to any other cluster within its light cone $\mathcal{L}(C_j)$:
\begin{equation}
    \mathcal{R}(C_j) = \max_{C_k \in \mathcal{L}(C_j)} \delta_{\text{inter}}(C_j, C_k).
\end{equation}

The \emph{cluster-level range of the unitary} $U$ is then defined as the maximum such range across all clusters in the system:
\begin{equation}
    \mathcal{R}(U) = \max_j \mathcal{R}(C_j).
\end{equation}
\end{definition}

These definitions extend the notion of causal propagation from individual qubits to clusters. To establish the physical basis for this cluster-level light cone, we relate it to the standard Lieb-Robinson (LR) bound. For a qubit lattice with local interactions, the LR theorem establishes an effective maximum speed for the spread of information, confining causal influence within a linear light cone up to exponentially small corrections~\cite{lieb1972finite}.

Our framework can be understood through the lens of \emph{coarse-graining}, where the system is viewed not as a lattice of qubits, but as a graph of clusters. In this picture, each cluster $C_j$ is treated as a single ``super-site''. The intra-cluster operations, $W_t$, act entirely within these super-sites, scrambling information locally but not transmitting it across cluster boundaries. 

In contrast, the inter-cluster operations $V_t$, constrained by geometric adjacency, act as the exclusive drivers of information propagation on this cluster graph. The causal structure emerges directly from the layered application of these gates. Starting from a trivial support $\mathcal{L}(C_j) = \{C_j\}$ at $t=0$, each subsequent layer $V_t$ expands the ``information front'' by at most one unit of inter-cluster distance $\delta_{\text{inter}}$. Consequently, after the full circuit of $T$ layers, the light cone $\mathcal{L}(C_j)$ is strictly confined to clusters within distance $T$ from the target. This implies a strict upper bound on the causal range:
\begin{equation}
    \mathcal{R}(U) = \max_j \mathcal{R}(C_j) \le T.
    \label{eq:range_bound}
\end{equation}
This bound defines an effective cluster-level Lieb-Robinson velocity $v_{\text{LR, cluster}} = 1$ (in units of clusters per layer), establishing the linear light cone used in our analysis.

This framework also accommodates non-geometric intra-cluster operations. In distributed quantum computing, for instance, each cluster may correspond to a highly connected processing unit (e.g., a transmon in a superconducting chip), while communication between clusters is mediated by tunable couplers enforcing limited geometric connectivity~\cite{chamberland2022modular}.

By extending the concepts of light cones and range to the cluster level, we can capture the localized propagation of quantum information both within and between clusters. This approach is crucial for simulating large-scale quantum systems on devices with limited resources, as it exploits spatial locality to reduce computational complexity and memory requirements.

To formalize this, we define $m$-sparse operators, qubit-level and cluster-level \(s\)-local operators. These definitions provide the mathematical framework for modeling and simulating quantum systems with varying degrees of locality and inter-cluster interaction. The precise definitions are presented below.

\begin{definition}[$m$-Sparse Operator]
\label{def:m_sparse_operator}
An operator $O$ is called \emph{$m$-sparse} if it can be expressed as a linear combination of at most $m$ Pauli strings. Throughout this work, we will represent such an operator in the form:
\begin{equation}
    O = \sum_{\alpha=1}^{m} c_\alpha P_\alpha,
\end{equation}
where $\{P_\alpha\}$ are Pauli strings and $\{c_\alpha\}$ are complex coefficients.

Associated with this decomposition is the 1-norm of the coefficients, denoted by $\lambda_c$, which is a key parameter that determines the sampling complexity for estimating the expectation value of $O$:
\begin{equation}
    \lambda_c = \sum_{\alpha=1}^{m} |c_\alpha|.
\end{equation}
\end{definition}

\begin{definition}[Cluster-Level \(s\)-local Operator]
A cluster-level operator \( O \) is called \emph{cluster-level \(s\)-local} if each term in its decomposition acts nontrivially on at most \(s\) clusters. That is, for every term \( P_\alpha \) in the decomposition of \( O \),
\[
\text{cluster-weight}(P_\alpha) \leq s \quad \forall\, \alpha.
\]
This means that the non-identity factors in each Pauli string are confined to at most \(s\) clusters.
\end{definition}

Finally, based on the concepts and definitions mentioned above, we provide a definition of the Pauli operator light cone and its volume.


\begin{definition}[Light Cone, Circuit, and Volume for a Pauli Term]
\label{def:pauli_light_cone}
Let $P_\alpha$ be a Pauli term in the observable $O$. The combined light cone, its associated minimal circuit, and its volume are constructed as follows:

\begin{enumerate}
    \item \textbf{Identify Support:} Identify the initial set of clusters $\mathcal{C}_\alpha = \{ C_j \mid P_\alpha \text{ acts non-trivially on } C_j \}$. By the $s$-local assumption, $|\mathcal{C}_\alpha| \le s$.

    \item \textbf{Construct Union of Light Cones:} For each cluster $C_j \in \mathcal{C}_\alpha$, determine its individual cluster light cone $\mathcal{L}(C_j)$. The combined light cone for the term $P_\alpha$ is the union of these individual cluster sets:
    \begin{equation}
        \mathcal{L}(\mathcal{C}_\alpha) = \bigcup_{C_j \in \mathcal{C}_\alpha} \mathcal{L}(C_j).
    \end{equation}

    \item \textbf{Construct Light Cone Circuit:} The minimal sub-circuit that correctly evolves $P_\alpha$, denoted $\widetilde{U}(\mathcal{C}_\alpha)$, is constructed by retaining only those gates (both intra- and inter-cluster) from the full unitary $U$ that satisfies
    \begin{equation}
        U^\dagger P_\alpha U=\widetilde{U}(\mathcal{C}_\alpha)^\dagger P_\alpha\widetilde{U}(\mathcal{C}_\alpha).
    \end{equation}

    \item \textbf{Define Light Cone Volume:} The light cone volume for the Pauli term $P_\alpha$, denoted $\text{Vol}(\mathcal{L}(\mathcal{C}_\alpha))$, is the total count of fundamental inter-cluster gates within its light cone circuit $\widetilde{U}(\mathcal{C}_\alpha)$.
\end{enumerate}
\end{definition}

\section{Main Results}
\label{sec:main}
The principle of locality, as formalized by the light cone structure in the previous section, offers a powerful lever for simulating complex quantum systems. The central challenge, however, lies in how to best exploit this structure. In this section, we explore two fundamentally different philosophies for tackling this challenge, leading to two distinct simulation algorithms.


Our first approach, the {Causal Decoupling Algorithm}, treats the light cone as a geometric graph. It seeks to find structural weaknesses—causal disconnections—that allow the quantum simulation itself to be broken into smaller, independent pieces. This strategy significantly reduces the sampling complexity but requires a quantum processor capable of handling these fractured, yet potentially large, sub-problems.

Our second approach, the {Algebraic Decomposition Algorithm}, takes this logic a step further. Instead of stopping at the geometric boundary of the light cone, it algebraically decomposes the gates \emph{within} the light cone itself. By expanding the entangling gates into a sum of local terms, it breaks down the large connected components required by the first method, entirely sidestepping the need for a multi-cluster quantum computer. This pushes the complexity into the classical post-processing and the number of measurement shots, enabling simulation on minimal hardware.

We now proceed to formalize these two frameworks, presenting their respective algorithms, performance-guaranteeing theorems, and practical complexity analyses.

\subsection{Causal Decoupling Algorithm}

We begin with the Causal Decoupling Algorithm, a geometric approach that directly utilizes the light-cone structure defined in Sec.~\ref{sec:bg}. Its core idea is to identify and leverage potential causal disconnections within the observable's light cone. If the light cone naturally fractures into disjoint components, the quantum simulation can be factorized into smaller, independent tasks without any algebraic modification to the gates.

Our result, formalized in Theorem~\ref{thm:Ddim}, establishes the resource profile for this geometric strategy. It demonstrates that the sample complexity scales polynomially with the observable's locality $s$, independent of the global system size.
\begin{theorem}[Local Observable on $D$-Dimensional Clustered Systems]\label{thm:Ddim}
Let a system be a $D$-dimensional clustered lattice (Def.~\ref{def:d_dim_lattice}) with maximum cluster size $d$. Let $U$ be a circuit with cluster-level depth $T$ and range $R$, and let $O = \sum_{\alpha=1}^{m} c_\alpha P_\alpha$ be an $m$-sparse, cluster-level $s$-local observable satisfying $\lambda_c=\sum_{\alpha=1}^{m}|c_\alpha| $.

For the input state $|0\rangle^{\otimes n}$, there exists a quantum algorithm that, for any error tolerance $\epsilon>0$ and with probability at least $2/3$, outputs an estimate $\widetilde\mu$ of $\mu = \langle 0^n|U^\dagger O U|0^n\rangle$ satisfying $|\widetilde\mu - \mu| \le \epsilon$.

The resource requirements of the algorithm are as follows:
\begin{itemize}
    \item \textbf{Sample Complexity:} $\mathcal{O}\left(\dfrac{\lambda_c^2 s^2}{\epsilon^2}\right)$
    \item \textbf{Qubit Requirement:} $\mathcal{O}(s d R^D)$
    \item \textbf{Circuit Depth:} $T$
\end{itemize}
\end{theorem}

To estimate this expectation value $\mu=\langle0^n|U^\dagger OU|0^n\rangle$, for this $s$-local observable $O=\sum_{\alpha=1}^mc_\alpha P_\alpha$, we propose the following Algorithm~\ref{alg:Ddim}.

\begin{algorithm}[Algorithm for Theorem~\ref{thm:Ddim}]
    \label{alg:Ddim}
    \begin{algorithmic}[1]
        \Require Observable $O = \sum_{\alpha=1}^m c_\alpha P_\alpha$, unitary $U$, error tolerance $\epsilon$.
        \Ensure An estimate $\widetilde{\mu}$ of $\mu = \bra{0^n}U^\dagger O U\ket{0^n}$ s.t. $|\widetilde{\mu} - \mu| \le \epsilon$.

        \State Initialize total estimate $\widetilde{\mu} \leftarrow 0$.
        \For{each Pauli term $P_\alpha$ in the decomposition of $O$}
            \State \textbf{Step 1: Construct Light Cone and Components}
            \State Construct the cluster light cone $\mathcal{L}(\mathcal{C}_\alpha)$ for $P_\alpha$ as per Definition~\ref{def:pauli_light_cone}.
            \State Identify the $k_\alpha$ disjoint connected components of $\mathcal{L}(\mathcal{C}_\alpha)$: $\{L_1, \dots, L_{k_\alpha}\}$.
            \State \Comment{This structure implies a factorization: $\mu_\alpha = \prod_{i=1}^{k_\alpha} \mu_{\alpha,i}$}
            
            \State \textbf{Step 2: Estimate Component Expectations}
            \State Initialize term estimate $\widetilde{\mu}_\alpha \leftarrow 1$.
            \For{each component $L_i$ from $i=1$ to $k_\alpha$}
                \State Define the local operator $P_{\alpha,i}$ (the part of $P_\alpha$ on $L_i$).
                \State Construct the corresponding light cone circuit $U_{\alpha,i}$ restricted to $L_i$.
                \State On a quantum processor, estimate $\mu_{\alpha,i} = \bra{0} U_{\alpha,i}^\dagger P_{\alpha,i} U_{\alpha,i} \ket{0}$ to get $\widetilde{\mu}_{\alpha,i}$.
            \EndFor
            
            \State \textbf{Step 3: Classical Aggregation for the Term}
            \State Compute the product $\widetilde{\mu}_\alpha = \prod_{i=1}^{k_\alpha} \widetilde{\mu}_{\alpha,i}$.
            \State Update the total estimate: $\widetilde{\mu} \leftarrow \widetilde{\mu} + c_\alpha \widetilde{\mu}_\alpha$.
        \EndFor
        \State \Return $\widetilde{\mu}$.
    \end{algorithmic}
\end{algorithm}

We provide the proof in Appendix.~\ref{sec:proof_thm_ddim}.



Theorem~\ref{thm:Ddim} confirms the efficacy of exploiting the light cone's geometry. The sample complexity is free from exponential scaling with circuit depth, determined instead by the polynomial parameter $s$.

However, the required qubit count $\mathcal{O}(s d R^D)$ may exceed modular device capacities for connected light cones. This motivates our second approach.

\begin{figure}[tb]
  \centering
  \includegraphics[width=\linewidth]{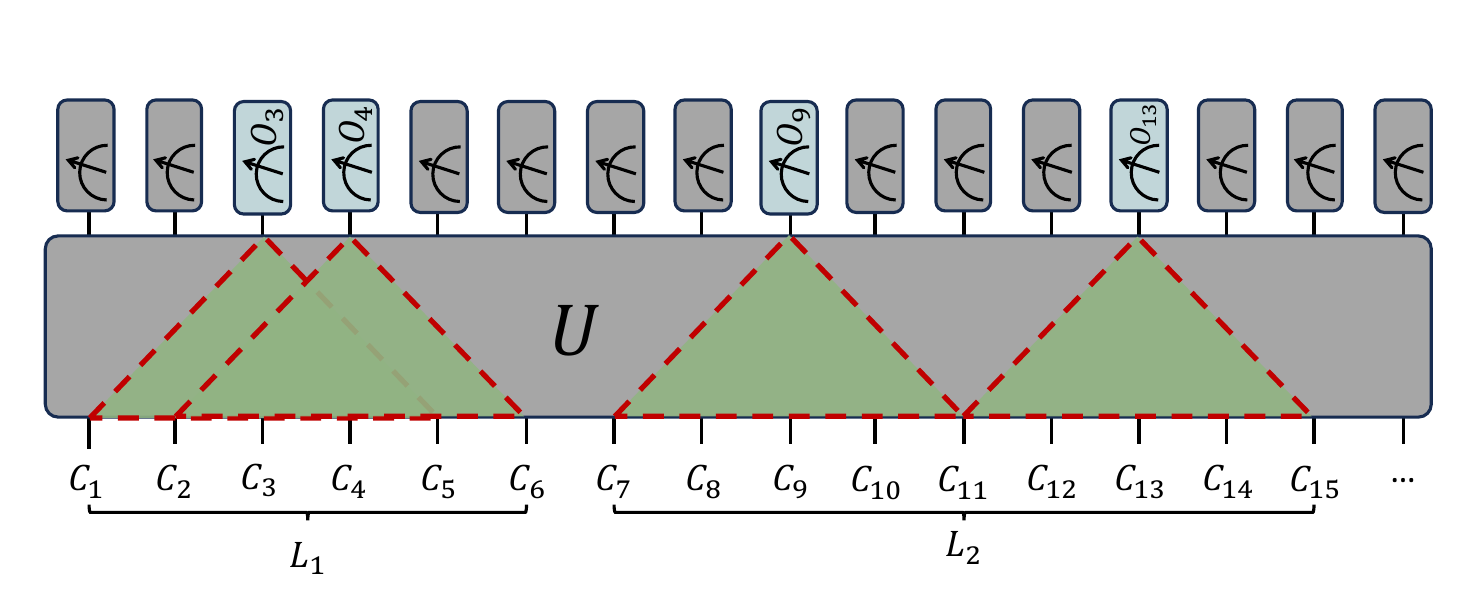}
  \caption{Illustration of the cluster‐level light‐cone decomposition for an \(s\)-local Pauli operator acting on a one‐dimensional chain of clusters under a shallow circuit \(U\).  The pale blue boxes mark the clusters \(\mathcal C_\alpha\) where the Pauli term acts non-trivially.  For each such cluster \(C_j\), its causal light cone \(L(C_j)\) under \(U\) is shown in green.  The union of overlapping light cones splits into connected supports \(L_1\) and \(L_2\), whose boundaries are outlined in red dashes.  We simulate each subcircuit \(U_{\mathrm{loc}}^{(\alpha,i)}\) by retaining only the gates within the corresponding support \(L_i\), since gates outside cancel between \(U\) and \(U^\dagger\).}
  \label{fig:cluster_light_cone}
\end{figure}

\subsection{Algebraic Decomposition Algorithm}
To address scenarios where the hardware resources are strictly limited—specifically, where one only has access to small, single-cluster devices—we introduce the Algebraic Decomposition Algorithm. This method takes the light-cone framework a step further. Instead of stopping at the geometric boundaries of the light cone components, it algebraically decomposes the inter-cluster gates \emph{within} the light cone itself.

We state our main result for this approach in Theorem~\ref{thm:clustered_systems}. This theorem establishes complexity bounds that trade the hardware requirements of the previous method for a sample complexity that depends on the volume of the light cone.

\begin{theorem}[Local Observable on Clustered Systems]
\label{thm:clustered_systems}
Let a system be composed of clusters with a maximum cluster size of $d$. Let $U$ be a quantum circuit with a cluster-level depth $T$, and let $O = \sum_{\alpha=1}^{m} c_\alpha P_\alpha$ be an $m$-sparse, cluster-level $s$-local observable satisfying $\lambda_c=\sum_{\alpha=1}^{m}|c_\alpha|$.

For the input state $\ket{0}^{\otimes n}$, there exists a quantum algorithm that, for any error tolerance $\epsilon>0$ and with probability at least $2/3$, outputs an estimate $\widetilde\mu$ of $\mu = \bra{0^n}U^\dagger O U\ket{0^n}$ satisfying $|\widetilde\mu - \mu| \le \epsilon$.

The resource requirements are determined by the properties of the largest light cone generated by any Pauli term in $O$. Let $\mathrm{Size}_{\max}$ be the maximum number of clusters in any such light cone, and let $\mathrm{Vol}_{\max}$ be its maximum volume (as defined in Def.~\ref{def:pauli_light_cone}):
\begin{align}
    \mathrm{Size}_{\max} &:= \max_{\alpha} |\mathcal{L}(\mathcal{C}_\alpha)|, \\
    \mathrm{Vol}_{\max} &:= \max_{\alpha} \mathrm{Vol}(\mathcal{L}(\mathcal{C}_\alpha)).
\end{align}
The requirements of the algorithm are as follows:
\begin{itemize}
    \item \textbf{Sample Complexity:}\\ $\mathcal{O}\left(\dfrac{\lambda_c^2 \cdot \mathrm{Size}_{\max}^2 \cdot 2^{\mathcal{O}(\mathrm{Vol}_{\max})}}{\epsilon^2}\right)$
    \item \textbf{Qubit Requirement:} $d+1$
    \item \textbf{Circuit Depth:} 
    {$2T+1$}
\end{itemize}
\end{theorem}

Algorithm~\ref{alg:gate_decomposition} achieves these bounds by leveraging the causal structure defined in Sec.~\ref{sec:bg}. As shown in Fig.~\ref{fig:gate_decom}, the strategy restricts evolution to the light cone and algebraically decomposes the entangling inter-cluster gates into a sum of local operations. This transforms the complex global evolution into a set of independent, single-cluster tasks solvable on minimal hardware.

\begin{figure}[tb]
    \centering
    \includegraphics[width=\linewidth]{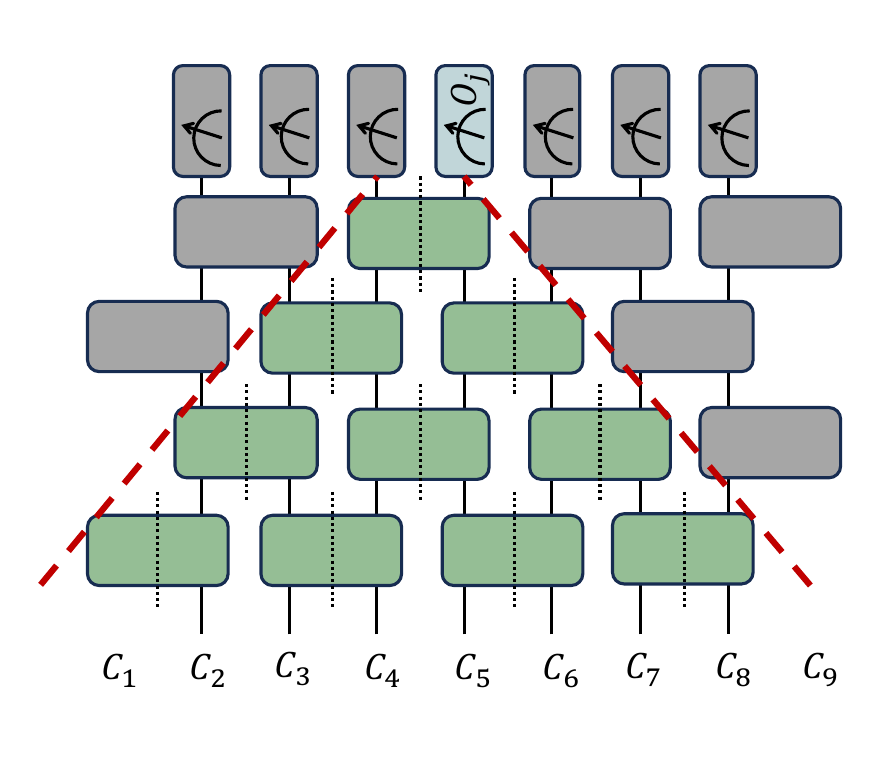}
    \caption{Causal light cone and task decomposition in a cluster-level brick-wall circuit. Each vertical line represents a quantum cluster $C_i$. For a local observable $O_j$ measured on a target cluster $C_j$ (here $C_5$) at the final time step, the support of its Heisenberg evolution, $U^\dagger O_j U$, is confined to a causal light cone (bounded by the red dashed line). Only the inter-cluster gates within this cone (green blocks) contribute to the measurement outcome, while gates outside of it (gray blocks) can be disregarded. Crucially, each gate within the light cone is decomposed, which allows the complex global evolution to be transformed into a series of local sub-problems. Each sub-problem can then be solved independently on a small quantum computer encompassing only the clusters within that light cone.}
    \label{fig:gate_decom}
\end{figure}

\begin{algorithm}[Algorithm via Inter-Cluster Gate Decomposition]
    \label{alg:gate_decomposition}
    \begin{algorithmic}[1]
        \Require Observable $O = \sum_{\alpha=1}^m c_\alpha P_\alpha$, unitary $U$, error tolerance $\epsilon$.
        \Ensure An estimate $\widetilde{\mu}$ of $\mu = \bra{0^n}U^\dagger O U\ket{0^n}$ s.t. $|\widetilde{\mu} - \mu| \le \epsilon$.

        \State Initialize total estimate $\widetilde{\mu} \leftarrow 0$.
        \For{each Pauli term $P_\alpha$ in the decomposition of $O$}
            \State \textbf{Step 1: Identify Light Cone Circuit}
            \State Construct the light cone $\mathcal{L}(\mathcal{C}_\alpha)$ and the corresponding light cone circuit $\widetilde{U}(\mathcal{C}_\alpha)$ for the term $P_\alpha$, as per Definition~\ref{def:pauli_light_cone}. Let $\text{Vol}_\alpha$ be the volume of this light cone.

            \State \textbf{Step 2: Decompose the Light Cone Circuit}
            \State Decompose every inter-cluster gate within $\widetilde{U}(\mathcal{C}_\alpha)$ into a linear combination of local operators. This transforms the unitary $\widetilde{U}(\mathcal{C}_\alpha)$ into a sum of circuits composed entirely of intra-cluster gates:
            \begin{equation*}
                \widetilde{U}(\mathcal{C}_\alpha) = \sum_{j=1}^{N_\alpha} w_j W_j
            \end{equation*}
            \State where each $W_j$ is a product of intra-cluster gates, and the number of terms $N_\alpha$ scales as $2^{\mathcal{O}(\text{Vol}_\alpha)}$.
            
            \State \textbf{Step 3: Evaluate Decomposed Expectation Value}
            \State The expectation $\mu_\alpha = \bra{0}\widetilde{U}(\mathcal{C}_\alpha)^\dagger P_\alpha \widetilde{U}(\mathcal{C}_\alpha)\ket{0}$ expands to $\sum_{j,k=1}^{N_\alpha} w_j^* w_k \bra{0} W_j^\dagger P_\alpha W_k \ket{0}$.
            \State Initialize term estimate $\widetilde{\mu}_\alpha \leftarrow 0$.
            \For{each pair of local circuits $(W_j, W_k)$}
                \State Let the term to estimate be $\eta_{jk} = \bra{0} W_j^\dagger P_\alpha W_k \ket{0}$, which factorizes over the $\text{Size}_\alpha$ clusters: $\eta_{jk} = \prod_{l=1}^{\text{Size}_\alpha} \mu_{jk,l}$, where $\mu_{jk,l} = \bra{\psi_{j,l}} P_{\alpha,l} \ket{\psi_{k,l}}$.
                \State Here {$\bra{\psi_{j,l}} P_{\alpha,l} \ket{\psi_{k,l}}$ is the matrix element of $P_{\alpha,l}$ (the restriction of $P_\alpha$ to cluster $l$) between states $\ket{\psi_{j,l}}$ and $\ket{\psi_{k,l}}$, prepared by $W_j$ and $W_k$ respectively.}
                \State Initialize product estimate $\widetilde{\eta}_{jk} \leftarrow 1$.
                \For{each cluster $l=1$ to $\text{Size}_\alpha$}
                    \State \Comment{Use Hadamard test on a small quantum processor with an ancilla qubit.}
                    \State \textbf{a) Estimate Real Part:}
                    \Statex \quad Construct the standard Hadamard test circuit for $\mu_{jk,l}$. This involves controlled applications of $V_{k,l}$ and $P_{\alpha,l}V_{j,l}$.
                    \Statex \quad Execute the circuit multiple times to estimate the ancilla's expectation value $\langle Z \rangle_{\text{anc}}$.
                    \Statex \quad Set $\widetilde{\text{Re}}(\mu_{jk,l}) \leftarrow \langle Z \rangle_{\text{anc}}$.
                    
                    \State \textbf{b) Estimate Imaginary Part:}
                    \Statex \quad Construct a modified Hadamard test circuit that includes an $S^\dagger$ gate on the ancilla before the final Hadamard gate.
                    \Statex \quad Execute this circuit to estimate the new ancilla expectation value $\langle Z \rangle'_{\text{anc}}$.
                    \Statex \quad Set $\widetilde{\text{Im}}(\mu_{jk,l}) \leftarrow \langle Z \rangle'_{\text{anc}}$.
                    
                    \State \textbf{c) Combine and Update:}
                    \Statex \quad Form the complex estimate $\widetilde{\mu}_{jk,l} = \widetilde{\text{Re}}(\mu_{jk,l}) + i \cdot \widetilde{\text{Im}}(\mu_{jk,l})$.
                    \Statex \quad Update the product: $\widetilde{\eta}_{jk} \leftarrow \widetilde{\eta}_{jk} \times \widetilde{\mu}_{jk,l}$.
                \EndFor
                \State Update the estimate for the Pauli term: $\widetilde{\mu}_\alpha \leftarrow \widetilde{\mu}_\alpha + w_j^* w_k \widetilde{\eta}_{jk}$.
            \EndFor

            \State \textbf{Step 4: Classical Aggregation}
            \State Update the total estimate: $\widetilde{\mu} \leftarrow \widetilde{\mu} + c_\alpha \widetilde{\mu}_\alpha$.
        \EndFor
        \State \Return $\widetilde{\mu}$.
    \end{algorithmic}
\end{algorithm}

The general complexity bounds in Theorem~\ref{thm:clustered_systems} depend on the abstract quantities $\text{Size}_{\max}$ and $\text{Vol}_{\max}$. To provide concrete performance estimates and understand the algorithm's behavior in practical scenarios, we now specialize these bounds to two paradigmatic architectures defined in our background section: geometrically local lattices and fully connected systems. These corollaries translate the abstract light cone properties into scaling laws dependent on physical parameters like circuit depth and system connectivity.

\begin{corollary}[Complexity on D-Dimensional Lattices]
\label{cor:lattice}
Consider a system arranged on a $D$-dimensional clustered lattice with a circuit of depth $T$ and cluster-level range $\mathcal{R}(U)=R$ (as per Def.~\ref{def:cluster_range}). We further assume that the number of elementary two-qubit gates applied between any pair of adjacent clusters per circuit layer is bounded by a constant. For an $s$-local observable, the sample complexity of Algorithm~\ref{alg:gate_decomposition} is then bounded by:
\begin{itemize}
    \item \textbf{Sample Complexity:}\\ $\mathcal{O}\left(\dfrac{\lambda_c^2 \cdot s^2 R^{2D} \cdot 2^{\mathcal{O}(sTR^{D})}}{\epsilon^2}\right)$
\end{itemize}
The qubit and depth requirements remain as stated in Theorem~\ref{thm:clustered_systems}. The proof is provided in Appendix~\ref{app:proof_corollaries}.
\end{corollary}

\begin{corollary}[Complexity on All-to-All Connected Systems]
\label{cor:all_to_all}
Consider a clustered system with all-to-all connectivity, and assume that the number of elementary gates applied between any pair of clusters per circuit layer is bounded by a constant. To characterize the growth of the operator support, we introduce a \emph{branching factor} $\beta > 1$. The support size for a single initial operator after $j$ layers is assumed to grow as $|s_j| \approx \min\{N, \beta^j\}$. For a circuit of depth $T$ and an $s$-local observable, the sample complexity of Algorithm~\ref{alg:gate_decomposition} is bounded by:
\begin{itemize}
    \item \textbf{Sample Complexity:} $\mathcal{O}\left(\dfrac{\lambda_c^2 \cdot s^2 \beta^{2T} \cdot 2^{\mathcal{O}(sT\beta^T)}}{\epsilon^2}\right)$
\end{itemize}
The qubit and depth requirements remain as stated in Theorem~\ref{thm:clustered_systems}. The proof is provided in Appendix~\ref{app:proof_corollaries}.
\end{corollary}

Theorem~\ref{thm:clustered_systems} quantifies the trade-off achieved by this deeper decomposition. By reducing the quantum hardware requirement to its absolute minimum—$\mathcal{O}(d)$ qubits, sufficient for a single cluster—we enable simulation on highly distributed architectures.

The cost for this architectural flexibility is a sample complexity that scales as $\mathcal{O}(2^{\mathcal{O}(\text{Vol}_{\max})})$. Unlike the Causal Decoupling method, this cost depends exponentially on the circuit depth (via the volume). However, for the important class of shallow-depth circuits, this overhead remains constant and manageable. This establishes the Algebraic Decomposition Algorithm as the viable strategy when quantum hardware is fragmented, complementing the geometry-based approach of the previous section.



\section{Discussion}
\label{sec:discuss}
\subsection{Interpretation of the Main Results}

Our two main theoretical results, presented in Theorem~\ref{thm:Ddim} and Theorem~\ref{thm:clustered_systems}, establish a framework for estimating local observables on large clustered quantum systems. The significance of this framework lies in its strategy to leverage the principle of causality to decouple the computational cost from the global system size $n$. Both algorithms presented achieve this by focusing computational resources exclusively within the causally relevant light cone of the observable.

Both algorithms demonstrate that the simulation cost depends on the light cone properties rather than the global state space.
Instead, the cost is determined by the properties of the observable's light cone—a region whose size is governed by the circuit's depth and connectivity, not the total number of qubits. This provides a rigorous guarantee that probing local properties of large quantum systems can be a feasible task, which is a requirement for applications from condensed matter physics to quantum chemistry.

While unified by this common approach, our two theorems reveal a trade-off in how the computational burden is distributed between quantum and classical resources. This trade-off is dictated by the depth to which the light cone structure is decomposed and the generality of the applicable architectures.

\paragraph{Causal Decoupling: Sampling Efficiency with Larger Quantum Hardware.}
Theorem~\ref{thm:Ddim} presents a strategy that prioritizes sampling efficiency by exploiting the geometric structure of the light cone. The sample complexity, $\mathcal{O}(\lambda_c^2 s^2/\epsilon^2)$, exhibits no exponential dependence on any circuit parameter. This efficiency relies on factorizing the simulation based on the light cone's connected components, a strategy most naturally applied to $D$-dimensional lattices where geometric locality is strictly defined. The cost is shifted to the quantum hardware: the qubit requirement, $\mathcal{O}(s d R^D)$, scales with the size of these components. This makes it a suitable approach for lattice systems when the light cone may be structurally sparse or disconnected, provided a sufficiently large quantum processor is available.

\paragraph{Algebraic Decomposition: Minimal Quantum Hardware with Sampling Overhead.}
In contrast, Theorem~\ref{thm:clustered_systems} quantifies a strategy that further decomposes the light cone algebraically. This approach is more general, applicable to any clustered architecture where a light cone volume can be defined, not just regular lattices. By decomposing the inter-cluster gates within the light cone, it reduces the qubit requirement to $\mathcal{O}(d)$, constant with respect to any light cone property. The trade-off for this hardware minimization and broader applicability is a sample complexity that scales as $\mathcal{O}(2^{\mathcal{O}(\text{Vol}_{\max})})$. This exponential dependence makes the method suitable for circuits with small light cone volumes, while its minimal footprint allows for simulation on highly distributed or irregular quantum architectures.

\subsection{Comparison with Other Methods}

To contextualize our main results, we compare our two protocols—the Causal Decoupling Algorithm (Theorem~\ref{thm:Ddim}) and the Algebraic Decomposition Algorithm (Theorem~\ref{thm:clustered_systems})—with the prominent technique of circuit cutting or quasi-probability simulation methods~\cite{peng2020simulating, mitarai2021constructing}. All three methods enable the simulation of a quantum system larger than the available hardware, yet they employ fundamentally different decomposition strategies, leading to distinct resource trade-offs.

We consider a canonical task: estimating the expectation value of an $s$-local observable on an $n$-qubit, $D$-dimensional clustered system after evolution by a circuit of depth $T$ and range $R$. The resource requirements for the three approaches are summarized in Table~\ref{tab:comparison}. For circuit cutting, we assume the system is partitioned along cluster boundaries.

\begin{table*}[ht]
    \centering
    \caption{Resource scaling comparison. (Common prefactors $\lambda_c^2/\epsilon^2$ are omitted from Sample Complexity for clarity.)}
    \label{tab:comparison}
    \renewcommand{\arraystretch}{1.5} 
    \begin{tabular}{l|c|c|c}
        \hline\hline
        \textbf{Resource} & \textbf{Gate Cutting~\cite{mitarai2021constructing}} & \textbf{Our Causal Decoupling} & \textbf{Our Algebraic Decomposition} \\
        \hline
        \textbf{Decomposition} & Spatial partitioning of the & Causal factorization of the & Algebraic expansion of the \\
        \textbf{Strategy} & global quantum state & light cone simulation & light cone evolution \\
        \hline
        \textbf{Qubit Requirement} & $d$ (per fragment) & $\mathcal{O}(s d R^D)$ & $d+1$ \\
        \hline
        \textbf{Circuit Depth} & $T$ & $T$ & 
        {$2T+1$} \\
        \hline
        \textbf{Sample Complexity} & $\mathcal{O}\left({N^2 \cdot 2^{\mathcal{O}(TD N)}}\right)$ & $\mathcal{O}\left({s^2}\right)$ & $\mathcal{O}\left({s^2 R^{2D} \cdot 2^{\mathcal{O}(sTR^D)}}\right)$ \\
        \hline
        \textbf{Dominant} & Exponential in global & Polynomial in observable & Exponential in light cone \\
        \textbf{Scaling Factor} & cluster count $N\sim n/d$ & locality $s$ & volume $\propto sTR^D$ \\
        \hline\hline
    \end{tabular}
\end{table*}
Table~\ref{tab:comparison} outlines a three-way trade-off between sampling overhead and quantum hardware requirements, which highlights the distinct characteristics of our two light-cone-based methods.

The primary challenge for circuit cutting is its exponential scaling with the global system size. The number of cuts required to partition the system grows with the boundary of the partitions. For a $D$-dimensional system of $N$ clusters, this boundary scales as $\mathcal{O}(T D (N-N^{\frac{D-1}{D}}))=\mathcal{O}(TDN)$, leading to a sample complexity that is exponential in a function of the total number of clusters, $N\sim n/d$. For any large-scale system where $n \gg d$, this overhead can become prohibitively large, potentially rendering the simulation intractable.

Both of our light-cone-based methods address this issue. By confining the simulation to the causal past of the observable, they eliminate any dependence on the global cluster count $N$ from the exponential factor in the sample complexity. The cost is instead tied to the properties of the light cone, which are governed by the circuit depth $T$ and range $R$, not the total system size $n$. This makes these protocols exponentially more efficient than circuit cutting for probing local physics in large-scale quantum systems.

Our two methods, in turn, offer a clear choice depending on the available quantum hardware and the specific circuit parameters. The Causal Decoupling algorithm provides a polynomial sample complexity, $\mathcal{O}(s^2\lambda_c^2/\epsilon^2)$. This is achieved at the cost of requiring a quantum processor large enough to accommodate an entire connected component of the light cone, with a qubit requirement of $\mathcal{O}(s d R^D)$.

The Algebraic Decomposition algorithm offers the inverse trade-off. It minimizes the quantum hardware requirement to $d+1$ qubits—sufficient for a single cluster—at the price of reintroducing an exponential scaling in the sample complexity. This scaling, $2^{\mathcal{O}(sTR^D)}$, depends on the light cone volume, not the global system size.

To illustrate the implications, we consider a scenario in 2D architectures where the number of qubits per cluster, $d$, scales as the square root of the total system size, $n$. This implies the number of clusters is $N\sim n/d = \mathcal{O}(\sqrt{n})$. We analyze the exponential factors for a shallow-depth circuit ($T, R, s = \mathcal{O}(1)$) in 2D ($D=2$):
\begin{itemize}
    \item \textbf{Circuit Cutting:} The exponential factor is $2^{\mathcal{O}(T D N)} = 2^{\mathcal{O}({N})} = 2^{\mathcal{O}(\sqrt{n})}$. The cost grows exponentially with the total system size.
    \item \textbf{Causal Decoupling:} There is no exponential factor in the sample complexity. The qubit requirement is $\mathcal{O}(s d R^D) = \mathcal{O}(\sqrt{n})$.
    \item \textbf{Algebraic Decomposition:} The exponential factor is $2^{\mathcal{O}(sTR^D)} = 2^{\mathcal{O}(1)}$. The sample complexity overhead is constant with respect to system size. The qubit requirement for each quantum task is $d+1 = \mathcal{O}(\sqrt{n})$.
\end{itemize}
This comparison shows a clear hierarchy. Circuit cutting becomes intractable for large systems due to its scaling. The choice between our two light-cone protocols depends on the available hardware architecture, not just the total number of available qubits.

The Causal Decoupling method's polynomial sample complexity is highly advantageous, but its hardware demand is stringent: it requires a single quantum processor capable of containing and executing an entire connected light-cone component, which in this scenario scales as $\mathcal{O}(\sqrt{n})$ qubits. This device must also support the specific inter-cluster connectivity required by the light-cone circuit.

Conversely, the Algebraic Decomposition method is characterized by its architectural flexibility. It does not require a single large processor. Instead, all quantum tasks are broken down into Hadamard tests performed on independent devices of size $\mathcal{O}(d) = \mathcal{O}(\sqrt{n})$ each, which need no connection between them.

Therefore, the decision is dictated by the available hardware architecture. If a large, connected quantum processor of size $\mathcal{O}(\sqrt{n})$ is available, the Causal Decoupling method is preferable due to its sampling efficiency. However, in a scenario with access to multiple, smaller, disconnected quantum modules, the Algebraic Decomposition method provides a viable path forward. This is a capability not offered by the other methods, which demand intra-light-cone connectivity.

\subsection{Practical Application Scenarios}

Our theoretical framework has direct applications in several domains of quantum simulation and algorithm benchmarking on near-term hardware.

The primary utility of the proposed protocols lies in tasks common in many-body physics and quantum chemistry, where Hamiltonians are typically sums of local terms~\cite{georgescu2014quantum,mcardle2020quantum}. A direct application is their use as a subroutine in Variational Quantum Algorithms (VQEs). To optimize the parameters of an ansatz, VQE requires repeated estimation of the Hamiltonian expectation value, $\langle H \rangle = \sum_\alpha c_\alpha \langle P_\alpha \rangle$. Our methods provide a way to compute each term $\langle P_\alpha \rangle$ with a cost that depends on the locality of the term itself, not the size of the entire system, which is a consideration for large-scale VQE simulations~\cite{peruzzo2014variational,tilly2022variational}.

Beyond energy estimation, our protocols are useful for characterizing properties of quantum states. They enable the direct computation of two-point (and higher-order) correlation functions, such as $\langle \sigma_i^z \sigma_j^z \rangle$, across varying distances. These correlators are fundamental for identifying quantum phases of matter, locating phase transitions, and understanding the propagation of information in dynamical systems~\cite{zeng2015quantum}.

The algorithms also find a natural application in quantum error correction (QEC). Benchmarking the performance of an encoded logical qubit can require measuring the expectation values of local stabilizer operators that define the QEC code~\cite{gottesman1997stabilizer,terhal2015quantum}. Our methods provide a resource-efficient approach to perform these checks on a logical qubit that may be part of a much larger quantum computation.

Finally, the favorable scaling of our local protocols highlights the inherent difficulty of estimating global observables. Properties like the fidelity of a state preparation, $\mathcal{F} = |\langle\psi|U|0^n\rangle|^2$, or the Loschmidt echo, $\mathcal{L}(t) = |\langle\psi_0|e^{-iHt}|\psi_0\rangle|^2$, are important for algorithm benchmarking and studying quantum chaos~\cite{flammia2011direct,peres1984stability,goussev2012loschmidt}. However, these quantities lack the local structure that our algorithms exploit. The efficiency of our methods for local tasks serves as a contrast, highlighting that a general-purpose solution for such global properties remains an open challenge, likely requiring resources that scale with parameters beyond just the observable's locality.

\subsection{Conclusion and Outlook}
\label{sec:conclusion}

In this work, we have developed and analyzed a framework for estimating the expectation values of local observables on large clustered quantum systems, based on the principle of the causal light cone. We presented two distinct algorithms—a Causal Decoupling method and an Algebraic Decomposition method—that both decouple the primary computational cost from the global system size. Together, they offer a versatile toolkit for quantum simulation on near-term modular architectures.

Our central contribution is the demonstration that probing local physics does not require simulating the entire, globally entangled quantum state. Our {Causal Decoupling Algorithm} (Theorem~\ref{thm:Ddim}) exploits the geometric structure of the light cone to achieve a polynomial sample complexity, $\mathcal{O}(\lambda_c^2 s^2/\epsilon^2)$, which is independent of any circuit or system size parameters. This efficiency, however, is contingent upon having a single, connected quantum processor large enough to accommodate an entire light-cone component, with a qubit requirement of $\mathcal{O}(sdR^D)$.

In contrast, our {Algebraic Decomposition Algorithm} (Theorem~\ref{thm:clustered_systems}) further decomposes the light cone algebraically to minimize quantum hardware demands. It requires only small, potentially disconnected quantum devices of size $\mathcal{O}(d)$, making it well-suited for highly distributed quantum architectures. This architectural flexibility is accompanied by a sample complexity that scales exponentially with the light-cone volume, $\mathcal{O}(2^{\mathcal{O}(\text{Vol}_{\max})})$. For shallow-depth circuits, this cost remains fixed and independent of the total system size.

These results delineate a clear trade-off between sampling complexity and quantum hardware capabilities, allowing practitioners to choose a suitable strategy based on available resources.

While our framework provides viable solutions, the analysis in this work is primarily tailored to regular lattices with nearest-neighbor interactions. Extending these protocols to more complex hardware topologies or physical systems with long-range interactions remains a relevant direction for future research~\cite{arute2019quantum,wright2019benchmarking,jurcevic2021demonstration}. Furthermore, the performance bounds could potentially be tightened, for instance by applying importance sampling to the classical summation in the Algebraic Decomposition algorithm~\cite{haferkamp2022linear}.

Future work could focus on validating these theoretical bounds through numerical benchmarking and experimental demonstrations on modular hardware. The inherent modularity of both algorithms makes them suited for integration with cluster-specific error mitigation techniques~\cite{li2017efficient,endo2018practical,temme2017error,cai2023quantum}. 
Finally, exploring synergies with classical simulation methods, such as using tensor network algorithms to contract the large classical sums generated by the Algebraic Decomposition algorithm, could lead to further efficiencies in hybrid quantum-classical computation~\cite{schollwock2011density}.

\begin{acknowledgements}
We thank Mile Gu and Yukun Zhang for fruitful discussions on early ideas of the project. 
This work is supported by the Innovation Program for Quantum Science and Technology (Grant No.~2023ZD0300200), the National Natural Science Foundation of China NSAF Grant No.~U2330201 and Grant No.~12361161602, Beijing Natural Science Foundation (Grant No.~Z250004), and Beijing Science and Technology Planning Project (Grant No.~Z25110100810000).
\end{acknowledgements}

\bibliographystyle{apsrev4-2}
\bibliographystyle{unsrt}
\bibliography{ref}

\appendix
\section{Proof of Theorem~\ref{thm:Ddim}}
\label{sec:proof_thm_ddim}

\subsection{Variance Bound for Product Estimators}
\label{app:lemma}
\begin{lemma}\label{lem:var-upper-bound}
Let $a_i := \mu_{\alpha,i}^2$ and $x_i := (1 - \mu_{\alpha,i}^2)/K_{\alpha,i}$ for $i=1,\dots,k_\alpha$.
Assume that $0 \le a_i \le 1$ and $K_{\alpha,i} \ge 1$, so that $x_i \ge 0$ and $a_i + x_i \le 1$.
Then the following inequality holds:
\begin{equation}\label{eq:appendix-var-bound}
\prod_{i=1}^{k_\alpha}\!\Big(a_i + x_i\Big)
- \prod_{i=1}^{k_\alpha}\! a_i
\;\le\;
\sum_{i=1}^{k_\alpha} x_i.
\end{equation}
Equivalently,
\begin{equation}\label{eq:var-bound-mu}
\prod_{i=1}^{k_\alpha}\!\Big(\mu_{\alpha,i}^2+\frac{1-\mu_{\alpha,i}^2}{K_{\alpha,i}}\Big)
-\prod_{i=1}^{k_\alpha}\!\mu_{\alpha,i}^2
\;\le\;
\sum_{i=1}^{k_\alpha}\!\frac{1-\mu_{\alpha,i}^2}{K_{\alpha,i}}.
\end{equation}
\end{lemma}

\begin{proof}
We prove the result by induction on the number of factors $k_\alpha$.

\paragraph{Base case.}
For $k_\alpha = 1$ the statement is an equality:
\[
(a_1 + x_1) - a_1 = x_1.
\]

\paragraph{Inductive step.}
Assume that the inequality~\eqref{eq:appendix-var-bound} holds for some integer $k_\alpha \ge 1$.
Define
\[
P_k := \prod_{i=1}^{k_\alpha} (a_i + x_i),
\quad
Q_k := \prod_{i=1}^{k_\alpha} a_i.
\]
For $k_\alpha + 1$ factors we have
\begin{align}
\prod_{i=1}^{k_\alpha+1}(a_i + x_i) - \prod_{i=1}^{k_\alpha+1} a_i
&= (a_{k_\alpha+1} + x_{k_\alpha+1}) P_k - a_{k_\alpha+1} Q_k \nonumber\\
&= a_{k_\alpha+1} (P_k - Q_k) + x_{k_\alpha+1} P_k. \label{eq:appendix-expansion}
\end{align}
By the induction hypothesis, $P_k - Q_k \le \sum_{i=1}^{k_\alpha} x_i$.
Using $0 \le a_{k_\alpha+1} \le 1$ and $P_k \le 1$ (since each $a_i + x_i \le 1$), we obtain
\[
a_{k_\alpha+1} (P_k - Q_k) \le \sum_{i=1}^{k_\alpha} x_i, \qquad
x_{k_\alpha+1} P_k \le x_{k_\alpha+1}.
\]
Substituting these bounds into Eq.~\eqref{eq:appendix-expansion} yields
\[
\prod_{i=1}^{k_\alpha+1}(a_i + x_i) - \prod_{i=1}^{k_\alpha+1} a_i
\le \sum_{i=1}^{k_\alpha+1} x_i,
\]
completing the induction.

\paragraph{Conclusion.}
Therefore, the inequality~\eqref{eq:appendix-var-bound} holds for all integers $k_\alpha \ge 1$.
Substituting back $a_i = \mu_{\alpha,i}^2$ and $x_i = (1 - \mu_{\alpha,i}^2)/K_{\alpha,i}$ gives the desired bound~\eqref{eq:var-bound-mu}.
\end{proof}

\paragraph{Remarks.}
The inequality is tight for $k_\alpha = 1$ and becomes strict whenever at least one $a_i < 1$ or one product term $(a_i+x_i) < 1$.
Intuitively, the left-hand side represents the cumulative variance contribution from multiplicative estimators, which is always bounded by the sum of their individual variances when the outcome variables are bounded in $[-1,1]$.

\subsection{Proof of Theorem~\ref{thm:Ddim}}
\begin{proof}[Proof of Theorem~\ref{thm:Ddim}]

The proof is structured in two parts. We first establish the correctness of the algorithm's factorization of the expectation value. We then analyze the variance of the resulting estimator to derive the resource costs.

\proofstep{1. Correctness of the Estimator.}
The algorithm targets the expectation value $\mu = \sum_{\alpha=1}^m c_\alpha \mu_\alpha$, where $\mu_\alpha = \langle 0^n|U^\dagger P_\alpha U|0^n\rangle$.

For a specific cluster-level $s$-local Pauli term $P_\alpha$, let $\mathcal{C}_\alpha$ be the set of at most $s$ clusters where it acts non-trivially. Due to the locality of the circuit $U$, the evolved operator $U^\dagger P_\alpha U$ is supported entirely within the union of the light-cones of the clusters in $\mathcal{C}_\alpha$. Let this union be $L(\mathcal{C}_\alpha) = \bigcup_{j \in \mathcal{C}_\alpha} L(C_j)$. Any gate in $U$ with support outside $L(\mathcal{C}_\alpha)$ commutes with $P_\alpha$ and cancels with its inverse from $U^\dagger$. The evolution can therefore be restricted to a local unitary $U_{\mathrm{loc}}^{(\alpha)}$ that only includes gates from $U$ acting on qubits within $L(\mathcal{C}_\alpha)$.

As described in Algorithm~\ref{alg:Ddim}, the support $L(\mathcal{C}_\alpha)$ may decompose into $k$ disjoint, causally disconnected components $\{L_1, \dots, L_{k_\alpha}\}$, where $1 \le k_\alpha \le s$. Consequently, both the local unitary and the Pauli operator factorize into a tensor product over these components:
\[
U_{\mathrm{loc}}^{(\alpha)} = \bigotimes_{i=1}^{k_\alpha} U_{\mathrm{loc}}^{(\alpha,i)} \quad \text{and} \quad P_\alpha = \bigotimes_{i=1}^{k_\alpha} O_{\alpha,i},
\]
where $O_{\alpha,i}$ comprises the parts of $P_\alpha$ supported on $L_i$. The initial state also factorizes. This leads to a full factorization of the expectation value:
\begin{align*}
\mu_\alpha 
&= \prod_{i=1}^{k_\alpha} \langle 0^{n_i} | (U_{\mathrm{loc}}^{(\alpha,i)})^\dagger O_{\alpha,i} U_{\mathrm{loc}}^{(\alpha,i)} | 0^{n_i} \rangle = \prod_{i=1}^{k_\alpha} \mu_{\alpha,i}.
\end{align*}
This justifies the core strategy of the algorithm: estimating the total expectation $\mu_\alpha$ by computing the product of independent, smaller expectation values $\mu_{\alpha,i}$. 

The final estimator for the term is $\tilde{\mu}_\alpha = \prod_{i=1}^{k_\alpha} \tilde{\mu}_{\alpha,i}$, where each $\tilde{\mu}_{\alpha,i}$ is the sample mean of $K_{\alpha,i}$ independent measurements. The expectation of each individual measurement shot is $\mu_{\alpha,i}$, making $\tilde{\mu}_{\alpha,i}$ an unbiased estimator for the component, i.e., $\mathbb{E}[\tilde{\mu}_{\alpha,i}] = \mu_{\alpha,i}$. Since the experiments for each component are independent, the random variables $\{\tilde{\mu}_{\alpha,i}\}$ are independent. Therefore, the expectation of their product is the product of their expectations:
\[
\mathbb{E}[\tilde{\mu}_\alpha] = \mathbb{E}\left[\prod_{i=1}^{k_\alpha} \tilde{\mu}_{\alpha,i}\right] = \prod_{i=1}^{k_\alpha} \mathbb{E}[\tilde{\mu}_{\alpha,i}] = \prod_{i=1}^{k_\alpha} \mu_{\alpha,i} = \mu_\alpha.
\]
This confirms that the estimator $\tilde{\mu}_\alpha$ is strictly unbiased.

\proofstep{2. Resource Cost Analysis.}
We now derive the resource complexities stated in Theorem~\ref{thm:Ddim}.

\paragraph{Sample Complexity.}
The estimator we construct is unbiased, i.e.\ $\mathbb{E}[\tilde{\mu}]=\mu$. Hence the mean squared error equals the variance,
\begin{equation}\label{eq:mse-var}
    \mathrm{MSE}(\tilde{\mu})=\mathrm{Var}(\tilde{\mu}).
\end{equation}
To guarantee the final estimation error $|\tilde{\mu}-\mu|\le\epsilon$ with the failure probability fixed at $\delta=\tfrac{1}{3}$ (as in Theorem~\ref{thm:Ddim}), it suffices to enforce the variance bound
\begin{equation}\label{eq:var-target}
    \mathrm{Var}(\tilde{\mu})\le \mathcal{O}(\epsilon^2).
\end{equation}


Write the observable as $O=\sum_{\alpha} c_{\alpha}P_{\alpha}$ and let $\tilde{\mu}=\sum_{\alpha} c_{\alpha}\tilde{\mu}_\alpha$ be the estimator built from per-term estimators $\tilde{\mu}_\alpha$. By independence of different Pauli terms' estimators (or by appropriate sampling design) the total variance decomposes as
\begin{equation}\label{eq:total-var-decomp}
    \mathrm{Var}(\tilde{\mu})=\sum_{\alpha} c_{\alpha}^2\,\mathrm{Var}(\tilde{\mu}_\alpha).
\end{equation}
For each $\alpha$ we form $\tilde{\mu}_\alpha$ as a product of $k_\alpha$ independent component estimators,
\begin{equation}\label{eq:product-estimator}
    \tilde{\mu}_\alpha=\prod_{i=1}^{k_\alpha}\tilde{\mu}_{\alpha,i},
\end{equation}
where each component $\tilde{\mu}_{\alpha,i}$ is the empirical mean of $K_{\alpha,i}$ independent $\{\pm1\}$-valued measurement outcomes with mean $\mu_{\alpha,i}$ and variance $(1-\mu_{\alpha,i}^2)/K_{\alpha,i}$. Independence across components implies the exact identity
\begin{equation}\label{eq:var-product-exact}
    \mathrm{Var}(\tilde{\mu}_\alpha)
    =\prod_{i=1}^{k_\alpha}\Big(\mu_{\alpha,i}^2+\frac{1-\mu_{\alpha,i}^2}{K_{\alpha,i}}\Big)
    -\prod_{i=1}^{k_\alpha}\mu_{\alpha,i}^2.
\end{equation}
While a first-order expansion for large $K_{\alpha,i}$ can provide intuition about the behavior of the variance, a stricter analysis yields a convenient upper bound that holds for any $K_{\alpha,i} \ge 1$. A formal proof, which proceeds by induction, is provided in Lemma~\ref{lem:var-upper-bound}. The bound is given by
\begin{equation}\label{eq:var-alpha-upper}
    \mathrm{Var}(\tilde{\mu}_\alpha)
    \le \sum_{i=1}^{k_\alpha}\frac{1-\mu_{\alpha,i}^2}{K_{\alpha,i}}
    \le \sum_{i=1}^{k_\alpha}\frac{1}{K_{\alpha,i}}.
\end{equation}
Combining \eqref{eq:total-var-decomp} and \eqref{eq:var-alpha-upper} gives
\begin{equation}\label{eq:total-var-upper}
    \mathrm{Var}(\tilde{\mu})
    \le \sum_{\alpha} c_{\alpha}^2 \sum_{i=1}^{k_\alpha}\frac{1}{K_{\alpha,i}}.
\end{equation}

We now choose the allocation of measurement shots $\{K_{\alpha,i}\}$ to minimize the total number of shots
\begin{equation}\label{eq:Ktot-def}
    K_{\rm tot}=\sum_{\alpha=1}^m \sum_{i=1}^{k_\alpha} K_{\alpha,i}
\end{equation}
subject to the variance constraint \eqref{eq:total-var-upper} being at most a target variance $V$. Consider the continuous relaxation and form the Lagrangian
\begin{equation}\label{eq:lagrangian}
    \mathcal{L}=\sum_{\alpha,i} K_{\alpha,i} + \lambda\Big(\sum_{\alpha} c_{\alpha}^2\sum_{i=1}^{k_\alpha}\frac{1}{K_{\alpha,i}}-V\Big),
\end{equation}
with multiplier $\lambda>0$. Stationarity $\partial\mathcal{L}/\partial K_{\alpha,i}=0$ yields
\begin{equation}\label{eq:Kalpha-solution}
    1-\lambda\frac{c_{\alpha}^2}{K_{\alpha,i}^2}=0
    \quad\Longrightarrow\quad
    K_{\alpha,i}=\sqrt{\lambda}\,|c_{\alpha}|.
\end{equation}
(Here we treat $K_{\alpha,i}$ as a continuous variable, as rounding to the nearest integer does not affect the dominant scaling behavior.)

Thus every component associated with the same Pauli term $\alpha$ receives the same number of shots, proportional to $|c_\alpha|$. Define
\begin{equation}\label{eq:K1-def}
    K_1:=\sum_{\alpha} k_\alpha |c_\alpha|.
\end{equation}
Using the constraint we obtain
\begin{equation}\label{eq:V-constraint}
    V = \sum_{\alpha} c_{\alpha}^2\sum_{i=1}^{k_\alpha}\frac{1}{\sqrt{\lambda}|c_\alpha|}
    = \frac{1}{\sqrt{\lambda}} \sum_{\alpha} k_\alpha |c_\alpha|
    = \frac{K_1}{\sqrt{\lambda}},
\end{equation}
hence $\sqrt{\lambda}=K_1/V$. The total number of shots then follows from \eqref{eq:Ktot-def} and \eqref{eq:Kalpha-solution}:
\begin{equation}\label{eq:Ktot-exact}
    K_{\rm tot}=\sum_{\alpha,i} K_{\alpha,i}=\sqrt{\lambda}\sum_{\alpha} k_\alpha |c_\alpha|
    =\sqrt{\lambda}\,K_1=\frac{K_1^2}{V}.
\end{equation}
Using the worst-case bound $k_\alpha\le s$ and $K_1\le s\sum_\alpha|c_\alpha|=s\lambda_c$ gives
\begin{equation}\label{eq:Ktot-worst}
    K_{\rm tot}=\mathcal{O}\!\left(\frac{\lambda_c^2 s^2}{V}\right).
\end{equation}


To obtain a high-probability guarantee with failure probability at most $\delta=\tfrac{1}{3}$ via Chebyshev's inequality, note that
\begin{equation}\label{eq:chebyshev}
    \Pr\big(|\tilde{\mu}-\mu|\ge \epsilon\big)\le \frac{\mathrm{Var}(\tilde{\mu})}{\epsilon^2}.
\end{equation}
Requiring the right-hand side to be $\le\delta$ is achieved by enforcing $\mathrm{Var}(\tilde{\mu})\le \delta\epsilon^2$, i.e.\ choosing $V=\delta\epsilon^2=(1/3)\epsilon^2$. Substituting this constant-$\delta$ choice into \eqref{eq:Ktot-worst} yields
\begin{equation}\label{eq:sample-complexity}
    K_{\rm tot}=\mathcal{O}\!\left(\frac{\lambda_c^2 s^2}{\delta\epsilon^2}\right)
    =\mathcal{O}\!\left(\frac{\lambda_c^2 s^2}{\epsilon^2}\right),
\end{equation}
where the last equality absorbs the constant factor $1/\delta=3$ into the big-$\mathcal{O}$ notation.


\paragraph{Qubit Requirement.}
The quantum subroutine simulates the evolution restricted to the union of light cones $L(\mathcal{C}_\alpha)$ associated with the clusters on which $P_\alpha$ acts. In a $D$-dimensional lattice a single cluster-level light cone of range $R$ covers at most $\mathcal{O}(R^D)$ clusters. Since each $P_\alpha$ acts on at most $s$ clusters, the union of these light cones contains at most $\mathcal{O}(sR^D)$ clusters. With each cluster containing up to $d$ physical qubits, the total number of qubits required to simulate any such light-cone region is therefore bounded by
\begin{equation}\label{eq:qubit-req}
    \mathcal{O}(s\,d\,R^D),
\end{equation}
which yields the quantum memory requirement stated in Theorem~\ref{thm:Ddim}.

\paragraph{Circuit Depth.}
Each local circuit $U_{\mathrm{loc}}^{(\alpha,i)}$ is obtained by restricting the global circuit $U$ to the corresponding light-cone region while preserving the temporal order of gates. Hence the depth of any local subcircuit cannot exceed the depth of the global circuit; consequently the circuit depth required by the protocol scales as
\begin{equation}\label{eq:circuit-depth}
    T,
\end{equation}
as claimed.

\end{proof}

\section{Proof Details for Theorem~\ref{thm:clustered_systems}}
\label{app:proof_gate_decomposition}


\begin{proof}[Proof of Theorem~\ref{thm:clustered_systems}]
The proof is structured in two main parts. We first establish the correctness of the estimator, which is derived from an algebraic decomposition of the light cone unitary. We then perform a rigorous variance analysis of this estimator to derive the stated resource costs.

\proofstep{1. Correctness of the Estimator.}
The algorithm targets the expectation value $\mu = \sum_{\alpha=1}^m c_\alpha \mu_\alpha$, where $\mu_\alpha = \langle 0^n|U^\dagger P_\alpha U|0^n\rangle$.

For a specific $s$-local Pauli term $P_\alpha$, its evolution is contained within the light cone circuit $\widetilde{U}(\mathcal{C}_\alpha)$. The core strategy of Algorithm~\ref{alg:gate_decomposition} is to decompose this unitary into a linear combination of operators that are local to individual clusters. This is a standard technique in quasi-probability simulations, yielding:
\begin{equation}
    \widetilde{U}(\mathcal{C}_\alpha) = \sum_{j=1}^{N_\alpha} w_j W_j,
\end{equation}
where each $W_j$ is a product of intra-cluster gates, the coefficients are constant-level~\cite{mitarai2021constructing}, $|w_j| = \mathcal{O}(1)$, and the number of terms $N_\alpha = 2^{\mathcal{O}(\text{Vol}_\alpha)}$. This decomposition allows us to express $\mu_\alpha$ exactly as:
\begin{equation}
    \mu_\alpha = \sum_{j,k=1}^{N_\alpha} w_j^* w_k \bra{0} W_j^\dagger P_\alpha W_k \ket{0} = \sum_{j,k=1}^{N_\alpha} w_j^* w_k \, \eta_{jk}.
\end{equation}
Each term $\eta_{jk}$ further factorizes into a product over the $\text{Size}_\alpha = |\mathcal{L}(\mathcal{C}_\alpha)|$ clusters within the light cone:
\begin{equation}
    \eta_{jk} = \prod_{l=1}^{\text{Size}_\alpha} \mu_{jk,l}, \quad \text{where} \quad \mu_{jk,l} = \bra{\psi_{j,l}} P_{\alpha,l} \ket{\psi_{k,l}}.
\end{equation}
The algorithm constructs an estimator $\tilde{\mu}_\alpha = \sum_{j,k} w_j^* w_k \tilde{\eta}_{jk}$. The estimator for each product term is $\tilde{\eta}_{jk} = \prod_l \tilde{\mu}_{jk,l}$. Each local component $\tilde{\mu}_{jk,l}$ is obtained from a Hadamard test on a small quantum processor dedicated to cluster $l$. Since these are separate quantum experiments performed on different clusters, the resulting estimators $\{\tilde{\mu}_{jk,l}\}_{l=1}^{\text{Size}_\alpha}$ are statistically independent random variables.

Each individual estimator $\tilde{\mu}_{jk,l}$ is unbiased, i.e., $\mathbb{E}[\tilde{\mu}_{jk,l}] = \mu_{jk,l}$. Due to the independence of these estimators, the expectation of their product is the product of their expectations:
\begin{equation}
    \mathbb{E}[\tilde{\eta}_{jk}] = \mathbb{E}\left[\prod_{l=1}^{\text{Size}_\alpha} \tilde{\mu}_{jk,l}\right] = \prod_{l=1}^{\text{Size}_\alpha} \mathbb{E}[\tilde{\mu}_{jk,l}] = \prod_{l=1}^{\text{Size}_\alpha} \mu_{jk,l} = \eta_{jk}.
\end{equation}
This confirms that $\tilde{\eta}_{jk}$ is an unbiased estimator for $\eta_{jk}$. Finally, by the linearity of expectation, the overall estimator $\tilde{\mu}_\alpha$ is also strictly unbiased.

\proofstep{2. Resource Cost Analysis.}
We now derive the resource complexities stated in the theorem.

\paragraph{Sample Complexity.}
The estimator $\tilde{\mu}=\sum_{\alpha} c_{\alpha}\tilde{\mu}_\alpha$ is unbiased, so its MSE equals its variance. To guarantee the final estimation error $|\tilde{\mu}-\mu|\le\epsilon$ with a constant failure probability, it suffices to enforce the variance bound $\mathrm{Var}(\tilde{\mu})\le \mathcal{O}(\epsilon^2)$. The total variance decomposes as:
\begin{equation}
    \mathrm{Var}(\tilde{\mu})=\sum_{\alpha} c_{\alpha}^2\,\mathrm{Var}(\tilde{\mu}_\alpha).
\end{equation}
For each $\alpha$, the estimator $\tilde{\mu}_\alpha = \sum_{j,k} w_j^* w_k \tilde{\eta}_{jk}$ is a sum over independently estimated terms. Thus, its variance is:
\begin{equation}
    \mathrm{Var}(\tilde{\mu}_\alpha) = \sum_{j,k=1}^{N_\alpha} |w_j^* w_k|^2 \mathrm{Var}(\tilde{\eta}_{jk}).
\end{equation}
We consider the estimator for each local term,
\begin{equation}
    \tilde{\eta}_{jk}=\prod_{l=1}^{\mathrm{Size}_\alpha}\tilde{\mu}_{jk,l},
\end{equation}
where the components $\tilde{\mu}_{jk,l}$ are independent complex estimators obtained from Hadamard tests.  
Independence implies that the variance admits the exact factorized form
\begin{equation}
    \mathrm{Var}(\tilde{\eta}_{jk})
    =\prod_{l=1}^{\mathrm{Size}_\alpha}\mathbb{E}\!\left[|\tilde{\mu}_{jk,l}|^{2}\right]
    -\prod_{l=1}^{\mathrm{Size}_\alpha}|\mu_{jk,l}|^{2}.
    \label{eq:var-product-one}
\end{equation}
Each complex estimator is obtained from two statistically independent Hadamard tests for the real and imaginary parts, each using $K_{jkl}/2$ samples. Denote these estimators by $\widetilde{\Re}$ and $\widetilde{\Im}$. The standard Hadamard-test statistics give
\begin{align}
    \mathbb{E}[\widetilde{\Re}] &= \Re(\mu_{jk,l}), &
    \mathrm{Var}(\widetilde{\Re}) &= \frac{1-\Re(\mu_{jk,l})^{2}}{K_{jkl}/2},\\
    \mathbb{E}[\widetilde{\Im}] &= \Im(\mu_{jk,l}), &
    \mathrm{Var}(\widetilde{\Im}) &= \frac{1-\Im(\mu_{jk,l})^{2}}{K_{jkl}/2}.
\end{align}
Since $\tilde{\mu}_{jk,l}=\widetilde{\Re}+i\,\widetilde{\Im}$ and the two parts are independent,
\begin{equation}
\begin{split}
    \mathbb{E}\!\left[|\tilde{\mu}_{jk,l}|^{2}\right]
    =\mathbb{E}[\widetilde{\Re}^{2}]+\mathbb{E}[\widetilde{\Im}^{2}]
    =|\mu_{jk,l}|^{2}
    +\frac{2-|\mu_{jk,l}|^{2}}{K_{jkl}/2}\\
    =|\mu_{jk,l}|^{2}+\frac{2(2-|\mu_{jk,l}|^{2})}{K_{jkl}}.
\end{split}
    \label{eq:second-moment}
\end{equation}
Because $P_{\alpha,l}$ is unitary, we have $|\mu_{jk,l}|^{2}\le 1$, and hence
\begin{equation}
\begin{split}
    \mathbb{E}\!\left[|\tilde{\mu}_{jk,l}|^{2}\right]
    =|\mu_{jk,l}|^{2}+\mathcal{O}\!\left(\frac{1}{K_{jkl}}\right),\\
    \mathrm{Var}(\tilde{\mu}_{jk,l})
    =\mathbb{E}\!\left[|\tilde{\mu}_{jk,l}|^{2}\right]-|\mu_{jk,l}|^{2}
    \le \mathcal{O}\!\left(\frac{1}{K_{jkl}}\right).
\end{split}
\end{equation}

To bound the variance of the product estimator, we use the standard inequality for independent random variables with bounded second moments:
\begin{equation}
\begin{split}
    \mathrm{Var}\!\left(\prod_{l=1}^{m} X_l\right)
    \le \left(\prod_{l=1}^{m}\mathbb{E}[X_l^{2}]\right)
    \sum_{l=1}^{m}
        \frac{\mathrm{Var}(X_l)}{\mathbb{E}[X_l^{2}]},\\
    X_l\ \text{independent}.
\end{split}
    \label{eq:product-variance-bound}
\end{equation}
Applying \eqref{eq:product-variance-bound} to $X_l=\tilde{\mu}_{jk,l}$ and using $\mathbb{E}[|\tilde{\mu}_{jk,l}|^{2}]\le 1+\mathcal{O}(1/K_{jkl})\le 2$ for all sufficiently large $K_{jkl}$, we obtain
\begin{equation}
    \mathrm{Var}(\tilde{\eta}_{jk})
    \le \mathcal{O}(1)\sum_{l=1}^{\mathrm{Size}_\alpha}\frac{1}{K_{jkl}},
\end{equation}
which captures the correct scaling with local cluster size and shot numbers.

Since the coefficients satisfy $|w_j|$ and $\sum_j|w_j|$ are both constant~\cite{mitarai2021constructing,singh2023experimental}, we combine these results to get the total variance bound:
\begin{equation}\label{eq:total-var-upper-alg1}
    \mathrm{Var}(\tilde{\mu}) \le \sum_{\alpha=1}^m c_\alpha^2 \left( C\cdot \sum_{j,k=1}^{N_\alpha} \sum_{l=1}^{\text{Size}_\alpha} \frac{1}{K_{jkl}} \right).
\end{equation}
We now choose the shot allocation $\{K_{jkl}\}$ to minimize the total number of shots, $K_{\text{tot}} = \sum_{\alpha,j,k,l} K_{jkl}$, subject to the variance constraint $\mathrm{Var}(\tilde{\mu}) \le V$, where $V=\mathcal{O}(\epsilon^2)$. We form the Lagrangian:
\begin{equation}
    \mathcal{L} = \sum_{\alpha,j,k,l} K_{jkl} + \lambda^\prime \left( C\cdot\sum_{\alpha} c_\alpha^2 \sum_{j,k,l} \frac{1}{K_{jkl}} - V \right).
\end{equation}
Stationarity, $\partial\mathcal{L}/\partial K_{jkl}=0$, yields the optimal allocation:
\begin{equation}
    1 - \lambda \frac{c_\alpha^2}{K_{jkl}^2} = 0 \quad \Longrightarrow \quad K_{jkl} = \sqrt{\lambda} |c_\alpha|,
\end{equation}
where $\lambda=\lambda^\prime C$.
This shows that the optimal number of shots for each local Hadamard test depends only on the coefficient $|c_\alpha|$ of its parent Pauli term. The total number of shots is:
\begin{equation}
\begin{split}
    K_{\text{tot}} = \sum_{\alpha,j,k,l} K_{jkl} = \sum_{\alpha=1}^m N_\alpha^2 \cdot \text{Size}_\alpha \cdot (\sqrt{\lambda}|c_\alpha|)\\
    = \sqrt{\lambda} \sum_{\alpha=1}^m |c_\alpha| N_\alpha^2 \text{Size}_\alpha.
\end{split}
\end{equation}
Now we use the variance constraint to find $\lambda$:
\begin{equation}
    V = C\sum_{\alpha} c_\alpha^2 \sum_{j,k,l} \frac{1}{\sqrt{\lambda}|c_\alpha|} = \frac{C}{\sqrt{\lambda}} \sum_{\alpha} |c_\alpha| N_\alpha^2 \text{Size}_\alpha.
\end{equation}
Solving for $\sqrt{\lambda}$ gives $\sqrt{\lambda} = \frac{C}{V} \sum_{\alpha} |c_\alpha| N_\alpha^2 \text{Size}_\alpha$. Substituting this into the expression for $K_{\text{tot}}$:
\begin{equation}
    K_{\text{tot}} = \frac{C}{V} \left( \sum_{\alpha=1}^m |c_\alpha| N_\alpha^2 \text{Size}_\alpha \right)^2.
\end{equation}
To obtain the final scaling, we take the worst-case bounds: $N_\alpha \le N_{\max} = 2^{\mathcal{O}(\text{Vol}_{\max})}$ and $\text{Size}_\alpha \le \text{Size}_{\max}$.
\begin{equation}
\begin{split}
    K_{\text{tot}} &\le \frac{C}{V} \left( \sum_{\alpha=1}^m |c_\alpha| N_{\max}^2 \text{Size}_{\max} \right)^2 \\
    &= \frac{C\cdot N_{\max}^4 \text{Size}_{\max}^2}{V} \left( \sum_{\alpha=1}^m |c_\alpha| \right)^2 \\
    &= \frac{C\cdot\lambda_c^2 \cdot \text{Size}_{\max}^2 \cdot (2^{\mathcal{O}(\text{Vol}_{\max})})^4}{V}\\
    &= \mathcal{O}\left( \frac{\lambda_c^2 \cdot \text{Size}_{\max}^2 \cdot 2^{\mathcal{O}(\text{Vol}_{\max})}}{V} \right).
\end{split}
\end{equation}
To obtain a high-probability guarantee via Chebyshev's inequality, we set the target variance $V = \mathcal{O}(\epsilon^2)$, yielding the final sample complexity:
\begin{equation}
    K_{\text{tot}} = \mathcal{O}\left( \frac{\lambda_c^2 \cdot \text{Size}_{\max}^2 \cdot 2^{\mathcal{O}(\text{Vol}_{\max})}}{\epsilon^2} \right).
\end{equation}

\paragraph{Qubit Requirement.}
The quantum subroutine involves estimating local transition amplitudes on a single cluster. This requires the qubits within that cluster (at most $d$) plus one ancillary qubit for the Hadamard test. The total number of qubits required for any quantum computation is therefore bounded by:
\begin{equation}
    d+1.
\end{equation}

\paragraph{Circuit Depth.}
Each local circuit $W_j$ is constructed by selecting gates from the global circuit $U$ that act within the light cone, while preserving their temporal sequence. Therefore, the depth of any circuit $W_j$, and consequently the depth of the controlled unitaries in the Hadamard tests, cannot exceed the depth of the global circuit. The circuit depth requirement is thus:
\begin{equation}
    T.
\end{equation}

\end{proof}

\section{Proofs for Corollaries}
\label{app:proof_corollaries}

Here we provide the proofs for the corollaries derived from Theorem~\ref{thm:clustered_systems} and its associated Algorithm~\ref{alg:gate_decomposition}. The core idea is to estimate the maximum light cone size and volume for specific system geometries and substitute these bounds into the general complexity formula.

\begin{proof}[Proof of Corollary~\ref{cor:lattice}]
We derive the complexity by specializing the general bounds from Theorem~\ref{thm:clustered_systems} to a $D$-dimensional lattice geometry. The key is to estimate the maximum light cone size, $\text{Size}_{\max}$, and volume, $\text{Vol}_{\max}$.

For an $s$-local operator $P_\alpha$, its support $\mathcal{C}_\alpha$ consists of at most $s$ clusters. On a lattice, the light cone of a single cluster after evolution with a circuit of range $R$ is a region of geometric radius $\mathcal{O}(R)$.

The scaling of the light cone size can be understood through a geometric argument based on dimensional recurrence. The number of clusters within a radius $R$ in a $D$-dimensional space is known to scale as $\mathcal{O}(R^D)$. Since the combined light cone $\mathcal{L}(\mathcal{C}_\alpha)$ is the union of at most $s$ such regions, its total size is bounded by:
\begin{equation}
    \text{Size}_{\max} = \max_\alpha |\mathcal{L}(\mathcal{C}_\alpha)| = \mathcal{O}(s R^D).
\end{equation}

The light cone volume, $\text{Vol}_{\max}$, is the total count of gates within this region integrated over the circuit's depth $T$. Assuming a roughly uniform distribution of gates, the volume is proportional to the depth multiplied by the average cone size:
\begin{equation}
    \text{Vol}_{\max} = \mathcal{O}(T \cdot \text{Size}_{\max}) = \mathcal{O}(sTR^D).
\end{equation}

Substituting this volume into the exponential term of the sample complexity from our main theorem gives the factor $2^{\mathcal{O}(sTR^D)}$. The polynomial prefactor, which scales with the square of the number of decomposed terms, can be bounded by the square of the light cone size, $\mathcal{O}(\text{Size}_{\max}^2) = \mathcal{O}(s^2 R^{2D})$.

Combining these parts, the total sample complexity becomes:
\begin{equation}
    \mathcal{O}\left(\dfrac{\lambda_c^2 \cdot s^2 R^{2D} \cdot 2^{\mathcal{O}(sTR^D)}}{\epsilon^2}\right).
\end{equation}
This completes the proof.
\end{proof}

\begin{proof}[Proof of Corollary~\ref{cor:all_to_all}]
We again start from the general complexity bound of Theorem~\ref{thm:clustered_systems} and specialize the light cone parameters for an all-to-all connected system, using the branching factor model.

In this model, the support of a single operator grows exponentially with the number of layers. For an $s$-local operator $P_\alpha$, which is initially supported on $s$ clusters, its support size after $T$ layers of the circuit will be approximately:
\begin{equation}
    \text{Size}_{\max} = \max_\alpha |\mathcal{L}(\mathcal{C}_\alpha)| \approx s \cdot \beta^T.
\end{equation}
The light cone volume $\text{Vol}_{\max}$ can be estimated as the product of the circuit depth $T$ and the final light cone size, representing a simplified time-integration. This gives:
\begin{equation}
    \text{Vol}_{\max} = \mathcal{O}(T \cdot \text{Size}_{\max}) = \mathcal{O}(sT\beta^T).
\end{equation}
This bound directly yields the exponential term $2^{\mathcal{O}(sT\beta^T)}$ in the sample complexity.

As established in the proof of Corollary~\ref{cor:lattice}, the polynomial prefactor is determined by the square of the maximum light cone size. Substituting the size for the all-to-all case, we get:
\begin{equation}
    \text{Prefactor} \propto \text{Size}_{\max}^2 = \mathcal{O}((s\beta^T)^2) = \mathcal{O}(s^2\beta^{2T}).
\end{equation}
Combining the prefactor, the exponential term, and the standard $\lambda_c^2/\epsilon^2$ scaling, the total sample complexity is:
\begin{equation}
    \mathcal{O}\left(\dfrac{\lambda_c^2 \cdot s^2 \beta^{2T} \cdot 2^{\mathcal{O}(sT\beta^T)}}{\epsilon^2}\right).
\end{equation}
This concludes the proof.
\end{proof}

\end{document}